\documentclass[pra,twocolumn, preprintnumbers,amsmath,amssymb]{revtex4}

\usepackage{graphicx, color, graphpap}% Include figure files
\usepackage{overpic}
\usepackage{amsmath}
\usepackage{enumitem}
\usepackage{amssymb}
\usepackage{amsthm}
\usepackage{pstricks}
\usepackage{float}
\usepackage{multirow}
\usepackage{hyperref}
\usepackage[T1]{fontenc}
\long\def\ca#1\cb{} %Use for commenting out: \ca...\cb

\newcommand{\ket}[1]{|#1\rangle}               %ket
\newcommand{\bra}[1]{\langle #1|}              %bra
\newcommand{\dya}[1]{\ket{#1}\!\bra{#1}}
\newcommand{\dyad}[2]{\ket{#1}\!\bra{#2}}        %dyad
\newcommand{\ip}[2]{\langle #1|#2\rangle}      %the inner product
\newcommand{\Cbb}{\mathbb{C}}

\newcommand{\Mbb}{\mathbb{M}}

\newcommand{\Wbb}{\mathbb{W}}
\newcommand{\Xbb}{\mathbb{X}}
\newcommand{\Ybb}{\mathbb{Y}}
\newcommand{\Zbb}{\mathbb{Z}}

\newcommand{\DC}{\mathcal{D}}
\newcommand{\EC}{\mathcal{E}}

\newcommand{\HC}{\mathcal{H}}

\newcommand{\MC}{\mathcal{M}}

\newcommand{\VC}{\mathcal{V}}

\newcommand{\Tr}{{\rm Tr}}

\renewcommand{\geq}{\geqslant}
\renewcommand{\leq}{\leqslant}

\newcommand{\mted}[3]{\langle#1|#2|#3\rangle }

\newcommand{\register}{Y}
\newcommand{\outcome}{y}

\newcommand{\VCt}{\widetilde{\mathcal{V}}}

\newcommand{\phiv}{\vec{\phi}}

\newcommand{\ot}{\otimes}
\newcommand{\ad}{^\dagger}
\newcommand*{\id}{\openone}

\newcommand*{\guess}{\text{guess}}

\newcommand{\rhoh}{\widehat{\rho}}

\newcommand{\rhot}{\tilde{\rho}}

\newcommand{\enh}{\text{enh}}
\newcommand{\secr}{\text{secr} }

\newcommand{\sys}{S}
\newcommand{\En}{E}
\newcommand{\eno}{E_1}
\newcommand{\ent}{E_2}
\newcommand{\fc}{\textsf{FC}}
\newtheoremstyle{example}{\topsep}{\topsep}%
{}%         Body font
{}%         Indent amount (empty = no indent, \parindent = para indent)
{\bfseries}% Thm head font
{:}%        Punctuation after thm head
{   }%     Space after thm head (\newline = linebreak)
{\thmname{#1}\thmnumber{ #2}}%\thmnote{ #3}}%         Thm head spec
\theoremstyle{example}
\newtheorem{example}{Example}%[subsection]

\newtheorem{theorem}{Theorem}
\newtheorem{lemma}[theorem]{Lemma}

\theoremstyle{definition}

\begin{document}

\title{Entropic framework for wave-particle duality in multipath interferometers}

\author{Patrick J. Coles}
\affiliation{Institute for Quantum Computing and Department of Physics and Astronomy, University of Waterloo, N2L3G1 Waterloo, Ontario, Canada}

\begin{abstract}
An interferometer - no matter how clever the design - cannot reveal both the wave and particle behavior of a quantum system. This fundamental idea has been captured by inequalities, so-called wave-particle duality relations (WPDRs), that upper bound the sum of the fringe visibility (wave behavior) and path distinguishability (particle behavior). Another fundamental idea is Heisenberg's uncertainty principle, stating that some pairs of observables cannot be known simultaneously. Recent work has unified these two principles for two-path interferometers. Here we extend this unification to $n$-path interferometers, showing that WPDRs correspond to a modern formulation of the uncertainty principle stated in terms of entropies. Furthermore, our unification provides a framework for solving an outstanding problem of how to formulate universally valid WPDRs for interferometers with more than two paths, and we employ this framework to derive some novel WPDRs.
\end{abstract}

\maketitle

\section{Introduction}\label{sec_intro}

Photons \cite{Schwindt1999}, electrons \cite{Bach2013}, neutrons \cite{Greenberger1988}, and even large organic molecules \cite{Arndt1999} have been shown experimentally to exhibit the behavior of both waves and particles. However, one cannot simultaneously see both behaviors; an apparatus that reveals the particle behavior cannot see the wave behavior and vice versa. This fundamental principle of quantum mechanics is known as wave-particle duality (WPD), and is closely related to Bohr's complementarity principle \cite{Scully1991, Bohr1928}.

Quantitative statements of WPD, so-called wave-particle duality relations (WPDRs), aim to upper bound the sum of the wave behavior and the particle behavior for a given interferometer. A well-known formulation given by Englert \cite{Englert1996} and Jaeger, Shimony, and Vaidman \cite{Jaeger1995} considered the two-path Mach-Zehnder interferometer (MZI) for single photons (see Fig.~\ref{fig1}, which shows more generally the $n$-path MZI). They quantified wave and particle behavior respectively by fringe visibility $\VC$ and path distinguishability $\DC$ (see below for definitions) and proved that
\begin{align}
\label{eqnVisDisTradeoff}
\VC^2+\DC^2 \leq 1.
\end{align}
This relation says that full wave-behavior ($\VC =1$) implies no particle behavior ($\DC =0$) and vice-versa, and also bounds the intermediate case of partial behavior. Many more complicated interferometry setups have been considered leading to a vast number of WPDRs \cite{Wootters1979a, Englert2000, Bosyk2013, Banaszek2013a, Li2012, Qureshi2013, Huang2013, Coles2014b, Vaccaro2012, AsadSiddiqui2015, Englert2007, Durr2001, Bera2015, Bagan2016, Qureshi2016}.

It has been debated \cite{Storey1994, Englert1995, Wiseman1995, Englert1996, Bjork1999, Durr2000, Busch2006a, Coles2014b} whether WPD is related to another fundamental principle - Heisenberg's uncertainty principle \cite{Heisenberg1927}. For example, the seminal paper Ref.~\cite{Englert1996} argued that the two principles are distinct. The uncertainty principle states that there exist pairs of observables, like position and momentum, that cannot be simultaneously known or jointly measured, and many quantitative statements of it have been formulated. Historically, these uncertainty relations employed the standard deviation as the uncertainty measure, but later they evolved into a more robust formulation in terms of entropies, i.e., entropic uncertainty relations (EURs). We refer the reader to \cite{Coles2015a} for a review of EURs and their applications to information-processing tasks such as cryptography.

Interestingly, several recent works \cite{Durr2000, Busch2006a, Coles2014b} demonstrated that the two principles are connected. Refs.~\cite{Durr2000, Busch2006a} connected \eqref{eqnVisDisTradeoff} to the standard-deviation-based uncertainty relation, while Ref.~\cite{Coles2014b} showed that \eqref{eqnVisDisTradeoff} is actually an EUR.  Hence, these works effectively unified two fundamental principles in quantum mechanics. 

However, a skeptic could argue that two-path interferometers are very special, and perhaps the equivalence between the two principles does not extend to arbitrary interferometers. This motivates our current work, where we unify the two principles for general $n$-path interferometers. In particular, we extend the result from Ref.~\cite{Coles2014b}, which found that \eqref{eqnVisDisTradeoff} is essentially the uncertainty relation for the min- and max-entropies, i.e., the relation used to prove the security of quantum key distribution~\cite{Tomamichel2012a}.

Furthermore, we exploit our aforementioned unification to provide a general framework for WPD in $n$-path interferometers. Such a framework has been lacking from the literature. On one hand, D\"urr \cite{Durr2001} proved a WDPR for the $n$-path MZI involving an operational measure of wave behavior, but at the cost of using a less operational measure of particle behavior.  On the other hand, Refs.~\cite{Bera2015, Bagan2016, Qureshi2016} recently proved WPDRs for this same scenario using operational measures of particle behavior, but at the cost of replacing visibility by coherence - a less operational measure of wave behavior. Here, we remedy this situation with measures of particle and wave behavior that are operational, experimentally friendly, and intuitive.

\begin{figure}[tbp]
\begin{center}
\includegraphics[width=3.4in]{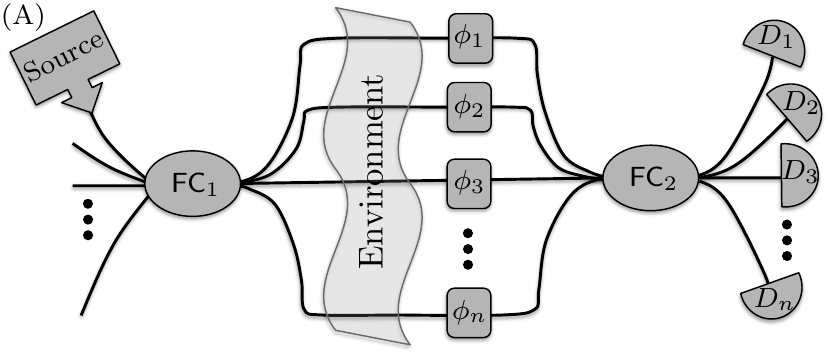}
\includegraphics[width=3.35in]{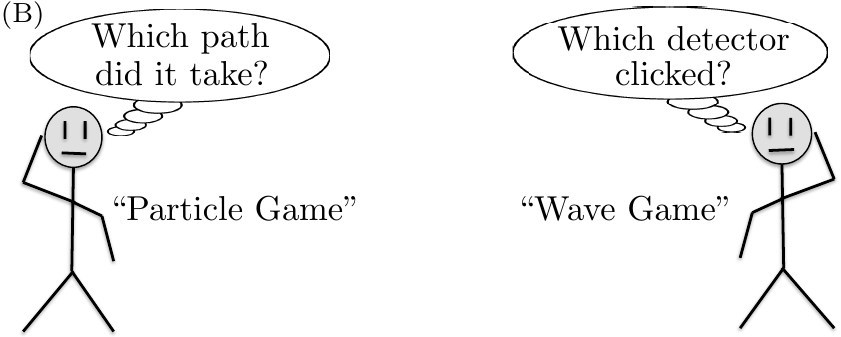}
\caption{(A) $n$-path interferometer for single photons. A source injects a photon into an optical fiber. The photon approaches a fiber coupler $\fc_1$, which allows the photon to leak into the other $n-1$ fibers, creating a superposition of which-path states. Then the photon interacts with some environment $E$, which may obtain some information, e.g., about the photon's path. Then a phase shift is applied to each arm ($\phi_z$ to the $z$-th arm), and the arms are brought together again at a second fiber coupler $\fc_2$. Finally the photon is detected at some detector. (B) Guessing game view of WPD. We think of ``particles'' as having a well-defined location, so the particle game asks the experimenter to guess (given access to a subsystem $E_1$ of $E$) which path the photon takes inside the interferometer. We think of ``waves'' as exhibiting interference - with the most extreme interference pattern corresponding to only one detector clicking all the time, while no interference corresponds to a uniform distribution over all detectors. Hence the wave game asks the experimenter to guess (given access to a subsystem $E_2$ of $E$, and given the freedom to adjust the phases $\{\phi_z\}$) which detector clicks. WPDRs give quantitative tradeoffs for the probabilities of winning these two games. }
\label{fig1}
\end{center}
\end{figure}

Our main conceptual results are as follows. We find a novel generalization of \eqref{eqnVisDisTradeoff} that extends the tradeoff between wave and particle behavior to $n$-path interferometers. Likewise we extend other WPDRs from the literature to the $n$-path case, namely, one treating quantum erasure~\cite{Englert2000} and one treating asymmetric input and output beam splitters \cite{Li2012}. The key point is that we derive these novel WPDRs directly from the uncertainty relation for the min- and max-entropies, which we henceforth abbreviate as MMEUR.  We argue that our WPDRs and many others in the literature are special cases of a single, generic WPDR [Eq.~\eqref{eqnMainResult} below] that is explicitly an EUR. In this sense we unify WPDRs with EURs.

The above conceptual insights come from some abstract, technical results that may be useful in other contexts. Namely we find a new connection between the max-entropy and the guessing probability, which leads to a new uncertainty relation for the guessing probability. We first present these technical results, and then we will move on to a discussion of wave-particle duality in Sec.~\ref{sec_general}. Finally, we compare of our approach to the literature in Sec.~\ref{sec_literature}.

\section{Abstract results}\label{sec_abstract}

The min- and max-entropies for a classical-quantum state $\rho_{XB} = \sum_x p_x \dya{x}\ot \rho^x_B$ are defined by \cite{Konig2009}\footnote{All logarithms are base 2.}
\begin{align}
\label{eqnminentropy}
H_{\min}(X|B)&:=-\log p_{\guess}(X|B)\\
\label{eqnmaxentropy}
H_{\max}(X|B)&:=\log p_{\secr}(X|B)\,.
\end{align}
Here,
\begin{align}
p_{\guess}(X|B) := \max_{\{\Mbb_x\}} \sum_x p_x \Tr (\Mbb_x \rho^x_B)
\end{align}
is the probability of guessing $X$ correctly given the outcome of the optimal POVM (positive operator valued measure) measurement $\{\Mbb_x\}$ on system $B$, and 
\begin{align}
\label{eqnpsecr}
p_{\secr}(X|B) := \max_{ \sigma_B } \textsf{F}(\rho_{XB}, \id \ot \sigma_B)
\end{align}
quantifies the secrecy of $X$ from $B$, as measured by the fidelity $\textsf{F}$ of $\rho_{XB}$ to an uncorrelated state, where the maximization in \eqref{eqnpsecr} is over all density operators on $B$.  Here, the fidelity is $\textsf{F}(\rho,\sigma):= (\Tr \sqrt{\sqrt{\rho}\sigma \sqrt{\rho}})^2$.

The min-entropy is employed in quantum key distribution to quantify how well the eavesdropper can guess the secret key. Interestingly we find that the \textit{max}-entropy is also connected to the guessing probability, with $d = |X|$,
\begin{align}
\label{eqnguessmaxentropy}
&H_{\max}(X|B) \notag\\
&  \leq \log \left(1+\sqrt{(d-1)^2 - (d \hspace{2pt}p_{\guess}(X|B)-1)^2   }\right)\,,
\end{align}
where the proof is in Appendix~\ref{app1}.

For any tripartite state $\rho_{AB_1B_2}$, and for two orthonormal bases $\Xbb=\{\dya{\Xbb_x}\}$ and $\Ybb=\{\dya{\Ybb_y}\}$ on $\HC_A$, whose measurement gives rise to random variables $X$ and $Y$, the MMEUR \cite{Tomamichel2011d} states that
\begin{align}
\label{eqnMMEUR}
H_{\min}(X|B_1)+H_{\max}(Y| B_2)\geq \log (1/c)
\end{align}
where $c:=\max_{x,y} |\ip{\Xbb_x}{\Ybb_y}|^2$. Inserting \eqref{eqnminentropy} and \eqref{eqnguessmaxentropy} into \eqref{eqnMMEUR} gives a novel uncertainty relation for $p_{\guess}$,
\begin{align}
\label{eqnpguessur1}
\frac{p_{\guess}(X|B_1)}{c}-\sqrt{(d-1)^2 - (d \hspace{2pt}p_{\guess}(Y|B_2)-1)^2   }\leq 1\,,
\end{align}
with $d=|X| = |Y|=\dim(\HC_A)$. In the extreme case of mutually unbiased bases (MUBs) we have $c = 1/d$, and \eqref{eqnpguessur1} simplifies to a symmetric looking relation
\begin{align}
\label{eqnpguessur2}
\DC(X|B_1)^2 + \DC(Y| B_2)^2 \leq 1 \,,
\end{align}
where $\DC(X|B):= (d \hspace{2 pt}p_{\guess}(X|B) - 1 ) / (d-1)$ is a measure of distinguishability with a range between 0 and 1. This concludes our abstract results, and we now move on to discuss WPD.

\section{Generic WPDR}\label{sec_general}

\subsection{$n$-path interferometer}\label{sec_npath}

Henceforth we consider an $n$-path interferometer. For example, Fig.~\ref{fig1} shows a source injecting a photon into a fiber optic, which then approaches a fiber coupler $\textsf{FC}_1$, creating a superposition of which-path states inside the interferometer.  The photon then interacts with an environment $E = E_1 E_2$, which (for generality) we allow to be a bipartite system composed of subsystems $E_1$ and $E_2$. Each path receives a phase shift $\phi_{z}$, with $\phiv:= \{\phi_z\}$ denoting the set of all phases. Finally, the paths are recombined with a fiber coupler $\textsf{FC}_2$, and the photon is detected at a detector.

\subsection{Guessing-game view}\label{sec_guess}

We argue that the MMEUR provides a robust, operational framework for discussing WPD, particularly due to the above connection with the guessing probability. Indeed one can think of WPD operationally as a statement that one cannot build an interferometer that allows one to win two complementary guessing games. Namely, as shown in Fig.~\ref{fig1}(B), we consider a ``particle game'' and a ``wave game''. The particle game asks the experimenter to guess which path the photon takes inside the interferometer, and the experimenter is given access to subsystem $E_1$ to help win the particle game. The wave game asks the experimenter to guess which detector will click at the interferometer output. To help win the wave game, the experimenter is given access to $E_2$ and furthermore is allowed to vary the phases $\phiv$. WPDRs are essentially quantitative restrictions on one's ability to win both the particle and wave guessing games.

To treat this general situation we state a generic WPDR. We quantify (lack of) particle and wave behavior by
\begin{align}
&\text{lack of particle behavior:} \hspace{4pt}H_{\min}(Z | \eno )\notag\\
&\text{lack of wave behavior:}\hspace{4pt} \min_{\Wbb \in \MC_{\Zbb}} H_{\max}(W | \ent ) \notag
\end{align}
where $Z$ is the which-path random variable, which we associate with the standard basis $\Zbb = \{\dya{z}\}$ of an $n$-dimensional Hilbert space $\HC_S$, $W$ is the random variable associated with basis $\Wbb$ of $\HC_S$, and $\MC_{\Zbb}$ is the set of all orthonormal bases that are mutually unbiased to $\Zbb$. We formulate our generic WPDR as
\begin{equation}
\label{eqnMainResult}
H_{\min}(Z | \eno )+ \min_{\Wbb\in \MC_{\Zbb}}H_{\max}(W | \ent ) \geq \log n
\end{equation}
which states that, for an $n$-path interferometer for single quantons (i.e., quantum particles such as photons), the sum of the ignorances about the particle and wave behaviors is at least $\log n$ bits of information. Of course,~\eqref{eqnMainResult} is explicitly an EUR, a special case of \eqref{eqnMMEUR}. But by applying the above argument used to derive \eqref{eqnpguessur2}, Eq.~\eqref{eqnMainResult} can be rearranged into the standard form for WPDRs, involving an upper bound on the sum of the squares of particle and wave terms. Furthermore, we now proceed to show that \eqref{eqnMainResult} encompasses several WPDRs in the literature and leads to novel WPDRs.

\section{Two paths}\label{sec_twopath}

First let us consider Eq.~\eqref{eqnVisDisTradeoff} for the two-path MZI, a special case of Fig.~\ref{fig1} where there is only a single relative phase $\phi$ applied between the two arms of the interferometer, i.e., $\phiv = \{0, \phi\}$. The path distinguishability quantifies how well one can guess which path the photon takes, given the outcome of the optimal measurement on the environment $E$ (e.g., $E$ could be the photon's polarization), and is defined by
\begin{align}
\label{eqnDistDef1}
\DC = 2 \hspace{1pt}p_{\guess}(Z|E)-1\,.
\end{align}
The fringe visibility $\VC$ quantifies oscillations, as one varies the phase $\phi$, in the probability to detect the photon at a given detector. Let $C\in \{1,2\}$ denote the random variable referring to which detector clicks at the interferometer output. Then $\VC$ is defined by
\begin{align}
\label{eqnVisDef1}
\VC &= \frac{p^{\max}_{C=1}-p^{\min}_{C=1}}{p^{\max}_{C=1}+p^{\min}_{C=1}}\,,
\end{align}
where $p^{\max}_{C=1} = \max_{\phiv}(p_{C=1})$, $p^{\min}_{C=1} = \min_{\phiv}(p_{C=1})$, and $p_{C=1}$ is the probability that $C=1$, i.e., that detector $D_1$ clicks. Ref.~\cite{Coles2014b} showed that $\DC$ and $\VC$ are respectively connected to the min- and max-entropies as follows
\begin{align}
\label{eqnMinMaxRelations1}
H_{\min}(Z|E) &= - \log \frac{1+\DC}{2}\\
\label{eqnMinMaxRelations2}
\min_{\Wbb \in \MC_{\Zbb}} H_{\max}(W) &= \log \Big(1+\sqrt{1-\VC^2}\Big) \,.
\end{align}
One can plug these identities into the EUR
\begin{align}
\label{eqnMinMaxEUR1}
H_{\min}(Z|E) + \min_{\Wbb \in \MC_{\Zbb}} H_{\max}(W) \geq 1
\end{align}
to show that \eqref{eqnVisDisTradeoff} is equivalent to \eqref{eqnMinMaxEUR1}. Note that \eqref{eqnMinMaxEUR1} is a special case of \eqref{eqnMainResult}, corresponding to $\eno = E$, $n=2$, and $\ent$ being a trivial system.

\section{Our WPDR for $n$ paths}\label{sec_novel}

To extend the above connection to the $n$-path MZI, we must address the question of how to generalize $\DC$ and $\VC$.  The seminal paper by Jaeger et al.~\cite{Jaeger1995} proposed that the appropriate generalization of $\DC$ to $n$ paths is:
\begin{align}
\label{eqnGeneralDistDef}
\DC = \frac{n \hspace{2pt}p_{\guess}(Z|E)-1}{n-1},
\end{align}
which reduces to \eqref{eqnDistDef1} for $n=2$.

The more difficult task is to generalize $\VC$. Naively, one might just try to directly use the formula in \eqref{eqnVisDef1}, in which case we denote the quantity as $\VCt$. However this approach fails.  As the following example illustrates, $\VCt$ does not satisfy a strong trade-off with $\DC$, for large $n$. 

\begin{example}
Consider the $n$-path MZI in Fig.~\ref{fig1}. Suppose $\fc_2$ induces the unitary transformation $F\ad$, where $F = \sum_{z,z'}(\omega^{-zz'}/\sqrt{n}) \dyad{z}{z'}$ is the Fourier matrix with $\omega = e^{2\pi i / n}$. Suppose the photon's state inside the interferometer is
\begin{align}
\label{eqnex1}\ket{\psi} &= \sum_{x\neq n}  \frac{1}{\sqrt{n-1}}\ket{\Xbb_x}\\
\label{eqnex2}&=\sqrt{\frac{n-1}{n}} \bigg( \ket{n}-\frac{1}{n-1}\sum_{z\neq n} \ket{z}\bigg)
\end{align}
where $\ket{\Xbb_x} = F \ket{x}$. From \eqref{eqnex2}, one can see that $p_{\guess}(Z) = (n-1)/n$ and hence $\DC = (n-2)/(n-1)$. So $\DC \to 1$ in the limit of large $n$. Also, $\VCt = 1$, since $p^{\max}_{C=1} > 0$ from \eqref{eqnex1}, and $p^{\min}_{C=1} = 0$ since choosing $\phi_z = \omega^z$ sets $p_{C=1} = 0$. Hence, for large $n$, one cannot formulate a non-trivial trade-off relation between $\DC$ and $\VCt$. 
\end{example}

While directly using \eqref{eqnVisDef1} fails, we seek to rewrite \eqref{eqnVisDef1} in a form that naturally generalize to $n$-paths. Assuming $\fc_2$ is symmetric (i.e., a photon with a well-defined path inside the interferometer has an equal chance to go into each of the output modes), note that for $n=2$
\begin{align}
\label{eqnvisalternative1}
\VC =  2p_{\guess}^{\max}(C)-1\,,
\end{align}
where $p_{\guess}^{\max}(C):=\max_{\phiv}p_{\guess}(C)$. Hence, for arbitrary $n$, we propose the following formula 
\begin{align}
\label{eqnWaveBehav1}
 \VC &:=  \frac{n \hspace{2pt}p_{\guess}^{\max}(C)-1}{n-1} \,.
\end{align}
Note the similarity between \eqref{eqnWaveBehav1} and \eqref{eqnGeneralDistDef}. We emphasize that \eqref{eqnWaveBehav1} captures the intuitive notion of interference visibility, since $p_{\guess}^{\max}(C)$ quantifies the spatial contrast of intensity at the interferometer output.

With these definitions we state the following result.
\begin{theorem}
\label{thm111}
For the $n$-path MZI in Fig.~\ref{fig1}, where $\fc_1$ is arbitrary while $\fc_2$ is symmetric, the generalization of Eq.~\eqref{eqnVisDisTradeoff} holds:
\begin{align}
\label{eqnVisDisTradeoff22}
\VC^2+\DC^2 \leq 1\,,
\end{align}
where $\VC$ and $\DC$ are defined in \eqref{eqnWaveBehav1} and \eqref{eqnGeneralDistDef}, respectively. Furthermore, \eqref{eqnVisDisTradeoff22} can be seen as a special case of the entropic uncertainty relation in \eqref{eqnMainResult}. 
\end{theorem}
\begin{proof}
The proof notes that~\eqref{eqnMinMaxRelations1} generalizes to 
\begin{align}
\label{eqnMinMaxRelations1b}
H_{\min}(Z|E) = - \log \left(\frac{1+(n-1)\DC}{n} \right)\,,
\end{align}
while \eqref{eqnMinMaxRelations2} generalizes with inequality:
\begin{align}
\label{eqnMinMaxRelations2b}
\min_{\Wbb \in \MC_{\Zbb}} H_{\max}(W) & \leq \log \left(1+(n-1)\sqrt{1-\VC^2} \right) \,,
\end{align}
where \eqref{eqnMinMaxRelations2b} follows from \eqref{eqnguessmaxentropy}, noting that 
\begin{align}
\label{eqnpguessWandC}
\max_{\Wbb \in \MC_{\Zbb}} p_{\guess}(W) = p_{\guess}^{\max}(C).
\end{align}
Inserting \eqref{eqnMinMaxRelations1b} and \eqref{eqnMinMaxRelations2b} into \eqref{eqnMainResult}, while setting $\ent$ to be trivial and $\eno = E$, and rearranging gives \eqref{eqnVisDisTradeoff22}.\end{proof}

\section{Extensions of our WPDR}\label{sec_extensions}

\subsection{Asymmetric couplers}\label{sec_assymetric}

The restriction in Theorem~\ref{thm111} that $\fc_2$ is symmetric can be relaxed, and a relation of the form of \eqref{eqnVisDisTradeoff22} can be obtained for the case where both $\fc_1$ and $\fc_2$ are possibly asymmetric. However, the price to pay is that one needs slightly more complicated definitions of visibility and distinguishability, where one post-selects on a particular detection event. This point was first highlighted by Ref.~\cite{Li2012} for the 2-path MZI. Here we extend the WPDR in Ref.~\cite{Li2012} to the $n$-path MZI, proving the relation (see Appendix~\ref{app3} for the proof and further discussion)
\begin{align}
\label{eqnasymmetricwpdr1}
\VC_1^2+\DC_1^2 \leq 1
\end{align}
where we define
\begin{align}
\label{eqnasymmetricwpdr2}
\VC_1 := \frac{p_{C=1}^{\max} - p_{C=1}^{\text{dec}}}{(n-1)p_{C=1}^{\text{dec}}}\,,\quad
\DC_1 := \frac{n \hspace{2pt}p_{\guess}(Z|E,C=1)-1}{n-1}\,.
\end{align}
Here, $\DC_1$ is the path distinguishability conditioned on the event $C=1$. Also, $p_{C=1}^{\text{dec}}$ denotes the probability for detector $D_1$ to click when the system's density matrix has been diagonalized (i.e., decohered) in the which-path basis. Note that $\VC_1$ quantifies the degree to which adjusting the phases $\phiv$ can increase the probability $p_{C=1}$ to detect the photon at detector $D_1$, beyond the baseline value $p_{C=1}^{\text{dec}}$ associated with no coherence. So it quantifies the effect of the applied phases on the final detection probability. Interestingly, the formula for $\VC_1$ somewhat resembles the standard definition for visibility in \eqref{eqnVisDef1} and, in fact, reduces to \eqref{eqnVisDef1} for $n=2$.

\subsection{Quantum erasure}\label{sec_qerase}

Finally, we show that \eqref{eqnMainResult} also provides a natural framework for a scenario called quantum erasure (see, e.g., \cite{Schwindt1999} for an experimental implementation), which aims to enhance the visibility by erasing the which-path information stored in the environment. Ref.~\cite{Englert2000} presented a WPDR for quantum erasure in two-path interferometers. Here we generalize their relation to the $n$-path case.

In quantum erasure, after the system $S$ interacts with an environment $E$, the experimenter performs a POVM measurement $\Ybb = \{\Ybb_y\}$ on $E$. This gives rise to sub-ensembles associated with the different measurement outcomes $y$, and one can define the distinguishability and visibility for the $y$-th sub-ensemble, which we respectively denote as $\DC(\Ybb_y)$ and $\VC(\Ybb_y)$. Averaging over all $y$ gives
\begin{align}
\label{eqndefdyvy}
\DC(\Ybb):= \sum_y p_y \DC(\Ybb_y)\quad \text{and}\hspace{5pt}\VC(\Ybb):= \sum_y p_y \VC(\Ybb_y)
\end{align}
where $p_y$ is the probability of outcome $y$. Our result for quantum erasure is that, for the $n$-path MZI in Fig.~\ref{fig1}, it holds for any choice of $\Ybb$ that~(see Appendix~\ref{app2} for proof)
\begin{align}
\label{eqnqerase123}
\VC(\Ybb)^2+ \DC(\Ybb)^2 \leq 1.
\end{align}
Note that \eqref{eqnqerase123} generalizes and implies our WPDR in \eqref{eqnVisDisTradeoff22}. One can recover \eqref{eqnVisDisTradeoff22} from \eqref{eqnqerase123} by noting that $\VC \leq \VC(\Ybb)$, and by choosing $\Ybb$ to optimize the distinguishability, since $\max_{\Ybb} \DC(\Ybb) = \DC$. Alternatively, one can choose $\Ybb$ to maximize $\VC(\Ybb)$, in which case \eqref{eqnqerase123} provides a fundamental limit on the recoverable visibility in a quantum erasure experiment.

We show (see Appendix~\ref{app2}) that \eqref{eqnqerase123} follows from a relation of the form
\begin{equation}
\label{eqnQuEr111122}
H_{\min}(Z|Y)+\min_{\Wbb \in \MC_{\Zbb}}H_{\max}(W|Y) \geq \log n,
\end{equation}
where $Y$ is the classical register that stores the outcome of the $\Ybb$ measurement on $E$. Note that \eqref{eqnQuEr111122} is a special case of \eqref{eqnMainResult} where $E_1 = Y$ and $E_2 = Y' $, where $Y'$ is a copy $Y$ (which has the same information content as $Y$, and hence can be replaced by $Y$ in the entropy term).

\section{Discussion of literature}\label{sec_literature}

The first WPDR for $n$-path interferometers was given by D\"urr \cite{Durr2001}. D\"urr exploited the fact that the purity function $\Tr \rho^2$ can be broken down into two terms:
\begin{equation}
\label{eqndurr}
 \sum_z \rho_{zz}^2 +\sum_{z,z'\neq z} |\rho_{zz'}|^2 = \Tr \rho^2\,,
\end{equation}
where $\rho_{zz'}:=\mted{z}{\rho}{z'}$. Since $\Tr \rho^2\leq 1$, Eq.~\eqref{eqndurr} gives a WPDR where the first and second terms on the left-hand side are interpreted as measures of particle and wave behavior, respectively. (More precisely, D\"urr incorporated some dimension-dependent scaling factors into these terms for normalization purposes.)

Furthermore, D\"urr generalized the relation in \eqref{eqndurr} to the more interesting case where a which-path detector obtains some information inside the interferometer, i.e., the scenario considered in \eqref{eqnVisDisTradeoff}. However, D\"urr's generalized relation has been critiqued \cite{Bimonte2003, Bimonte2003a, Jakob2007} due to the fact that it is not saturated for all pure states. Nevertheless, we do not see this as a major issue. Consider that Maassen and Uffink's EUR \cite{Maassen1988} is not saturated by all pure states. Yet their relation is by far the most famous EUR, and it has inspired countless studies on the topic.

We believe that much more important issues are whether the employed quantitative measures are (1) operational / experimentally friendly and (2) capture one's intuition about wave and particle behavior. One could argue that D\"urr was on the right track, in this sense, with his measure of wave behavior, for which he gave an operational interpretation \cite{Durr2001}. However, D\"urr's measure of particle behavior is not as operational or intuitive as the one in~\eqref{eqnGeneralDistDef}, proposed by Ref.~\cite{Jaeger1995}.

Interestingly, some recent studies \cite{Bera2015,Bagan2016,Qureshi2016} seem to do the opposite of D\"urr. They employ operational measures of particle behavior, e.g., Ref~\cite{Bagan2016} employs the definition in \eqref{eqnGeneralDistDef}. But these references replace visibility with a more abstract quantity called coherence. While coherence is an interesting quantity for theorists, it remains to be clarified how it precisely relates to interferometry experiments.

The nice aspect of our approach is that both the wave and particle terms are operational - both involving the optimal guessing probability. The symmetric nature of the two terms seems quite natural. It leads to a simple guessing-game view of wave-particle duality, as shown in Fig.~\ref{fig1}(B). 

Furthermore, all of the previous works stop short of considering the case where both the input and output beam splitters are asymmetric. Indeed the relations derived in \cite{Durr2001,Bera2015,Bagan2016,Qureshi2016} do not apply to this situation. The fact that we can treat this situation is a testament to the robustness of our entropic-uncertainty framework.

We remark that, after completion of our work, Renes \cite{Renes2016} proved an uncertainty relation that is analogous to (but not exactly the same as) Eq.~\eqref{eqnpguessur2}. In some cases Renes's bound is tighter than \eqref{eqnpguessur2}, and in other cases the opposite is true. Interestingly, like us, Renes derived his relation directly from the MMEUR.

\section{Conclusion}\label{sec_conclusion}

In summary, we obtained three novel WPDRs for $n$-path interferometers, Eqs.~\eqref{eqnVisDisTradeoff22}, \eqref{eqnasymmetricwpdr1}, and \eqref{eqnqerase123}, each of which generalize the famous WPDR of Refs.~\cite{Englert1996,Jaeger1995}. All of these novel WPDRs follow directly from the MMEUR, combined with our new operational meaning for the max-entropy in \eqref{eqnguessmaxentropy}. In this sense, wave-particle duality is the entropic uncertainty principle in disguise, and the latter provides a robust framework for formulating the former. We emphasize that our generic WPDR in \eqref{eqnMainResult} can be applied to a variety of interferometric scenarios, and hence, when specialized, will lead to other novel WPDRs.

\section{Acknowledgements}\label{sec_ack}

The author is funded by Sandia National Laboratories, Office of Naval Research, Industry Canada, NSERC Discovery Grant, and Ontario Research Fund.

\bibliographystyle{naturemag}

\newpage

\onecolumngrid

\appendix

\section{Relation between max-entropy and guessing probability}\label{app1}

Here we prove \eqref{eqnguessmaxentropy}. First we state a technical lemma involving the $p$-norm, which is defined by
\begin{align}
\label{eqnpnorm}
\| \vec{x} \|_{p} :=  \bigg(\sum_j |x_j |^p \bigg)^{1/p}
\end{align}
for a real vector $\vec{x} = \{x_j\}$. In particular, we consider the special cases
\begin{align}
\label{eqnpnorm2}
\| \vec{x} \|_{\infty} &= \max_j |x_j | \,,\\
\| \vec{x} \|_{1/2} &= \bigg(\sum_j \sqrt{|x_j | } \bigg)^2\,.
\end{align}
We remark that $\| \vec{x} \|_{p}$ is technically a norm for $p\geq 1$, while this is not true for $0<p<1$.

\begin{lemma}
\label{lemma1}
For a discrete probability distribution $\vec{q} = \{q_j\}$ over a sample space of size $d$, it holds that
\begin{align}
\label{eqnnormlemma}
\big( \| \vec{q} \|_{1/2} - 1\big)^2 + \big(d \hspace{2pt} \| \vec{q} \|_{\infty} - 1 \big)^2 \leq \big(d-1 \big)^2\,,
\end{align}
where equality holds if $d=2$.
\end{lemma}
\begin{proof}
The function $\| \vec{q} \|_{1/2}$ is Schur-concave, i.e., it satisfies 
\begin{align}
\label{eqnnormlemma1}
\| \vec{q} \|_{1/2}\leq \| \vec{r} \|_{1/2},\quad\text{if  }\vec{q} \succ \vec{r} \,.
\end{align}
The majorization condition $\vec{q} \succ \vec{r}$ is defined by
\begin{align}
\label{eqnnormlemma2}
\sum_{j=1}^k q_j \geq \sum_{j=1}^k r_j \quad \text{for all }k \in [d]\,,
\end{align}
where the probabilities are assumed to be listed in descending order, i.e., $q_j \geq q_l$ for $j<l$ and likewise for $\vec{r}$.

Now let us choose
\begin{align}
\label{eqnnormlemma3}
\vec{s} := \left\{q_1, \frac{1-q_1}{d-1}, \frac{1-q_1}{d-1},... , \frac{1-q_1}{d-1} \right\}\,, 
\end{align} 
and note that this choice gives $\vec{q} \succ \vec{s}$. Hence we have
\begin{align}
\label{eqnnormlemma4}
\| \vec{q} \|_{1/2} \leq \| \vec{s} \|_{1/2}= \Big(\sqrt{q_1} + \sqrt{(d-1)(1- q_1)  }\Big)^2\,.
\end{align} 

Now consider the function
\begin{align}
\label{eqnnormlemma5}
f(q_1):= (d-1)^2  - (d q_1 - 1)^2 - (\| \vec{s} \|_{1/2} - 1)^2 \,.
\end{align} 
We wish to show that $f(q_1) \geq 0$. Note that, using \eqref{eqnnormlemma4} and the fact that $ \| \vec{q} \|_{\infty} = q_1$, the non-negativity of $f(q_1) $ would imply that the desired result \eqref{eqnnormlemma} is true.

The equality of \eqref{eqnnormlemma} for $d=2$ is easily verified, so in what follows we restrict to $d\geq 3$. To show $f(q_1) \geq 0$ we write
\begin{align}
\label{eqnnormlemma6}
f(q_1) = (1- q_1 )(a(q_1)-b(q_1)) \,,
\end{align} 
where
\begin{align}
\label{eqnnormlemma7}
a(q_1)&:=d(d-2) +d^2 q_1\\
 b(q_1)&:=\Big( (d-2)\sqrt{1-q_1} +\sqrt{q_1(d-1)}\Big)^2\,.
\end{align} 

We just need to verify that $a(q_1) \geq b(q_1)$. It is straightforward to show that 
\begin{align}
\label{eqnnormlemma8}
\max_{q_1 \in [0,1]} b(q_1) = (d-2)^2 + d-1\,,
\end{align} 
and clearly $a(q_1)\geq d(d-2)$, so we have
\begin{align}
\label{eqnnormlemma9}
a(q_1) -  b(q_1) \geq  d-3 \,,
\end{align} 
which proves $f(q_1)\geq 0$ and hence \eqref{eqnnormlemma}.
\end{proof}

We note that \eqref{eqnnormlemma} is equivalent to the following
\begin{align}
\label{eqnguessmaxentropynomem}
H_{\max}(X)   \leq \log \left(1+\sqrt{(d-1)^2 - (d \hspace{1pt}p_{\guess}(X)-1)^2   }\right)\,,
\end{align}
which is a special case of relation \eqref{eqnguessmaxentropy}, corresponding to a trivial $B$ system. Now we prove the general case, where $B$ is non-trivial. For convenience, we restate \eqref{eqnguessmaxentropy} in the following lemma.

\begin{lemma}
\label{lemma2}
For a classical-quantum state $\rho_{XB} = \sum_x p_x \dya{x}\ot \rho_B^x$ where $d = |X|$, it holds that
\begin{align}
\label{eqnguessmaxentropyapp}
H_{\max}(X|B)  \leq \log \left(1+\sqrt{(d-1)^2 - (d \hspace{2pt}p_{\guess}(X|B)-1)^2   }\right)\,.
\end{align}
\end{lemma}
\begin{proof}
In what follows, we use some properties of the min- and max-entropies, and we refer the reader to Ref.~\cite{Tomamichel2012} for elaboration. For example, the data-processing inequality for the max-entropy implies that
\begin{align}
\label{eqnguessmaxentropyapp2}
2^{H_{\max}(X|B)} \leq 2^{H_{\max}(X|M)}
\end{align}
where $M$ is the classical register produced from measuring POVM $\Mbb = \{ \Mbb_m\}$ on system $B$. That is, the right-hand side of \eqref{eqnguessmaxentropyapp2} is evaluated for the state
\begin{align}
\label{eqnguessmaxentropyapp2b}
\rho_{XM}:= \sum_m \Tr_B \left(\rho_{XB} (\id \ot \Mbb_m)\right) \ot \dya{m} = \sum_{x,m} \Tr(p_x\rho_B^x \Mbb_m) \dya{x}\ot \dya{m}\,.
\end{align}

Suppose we choose $\Mbb$ such that it is the measurement that optimizes the guessing probability, i.e., 
\begin{align}
\label{eqnguessmaxentropyapp3}
p_{\guess}(X|B)&= \sum_x p_x \Tr (\Mbb_x \rho_B^x)\,.
\end{align}
Using properties of the guessing probability, we can write 
\begin{align}
\label{eqnguessmaxentropyapp3b}
p_{\guess}(X|B)&= p_{\guess}(X|M) = \sum_m q_m p_{\guess}(X|M=m)\,,
\end{align}
where $q_m := \Tr (\rho_B \Mbb_m)$ is the probability of outcome $M=m$, and $p_{\guess}(X|M=m)$ denotes the guessing probability for $X$ conditioned on outcome $M=m$. 
 
From \eqref{eqnguessmaxentropyapp2}, and the expression for the max-entropy when conditioning on classical information \cite{Tomamichel2012}, we have 
\begin{align}
\label{eqnguessmaxentropyapp4}
2^{H_{\max}(X|B)} &\leq \sum_m q_m 2^{H_{\max}(X|M=m)}\,, 
\end{align}
where $H_{\max}(X|M=m)$ is the max-entropy of $X$ conditioned on outcome $M=m$. Combining \eqref{eqnguessmaxentropyapp4} with \eqref{eqnguessmaxentropynomem} gives
\begin{align}
2^{H_{\max}(X|B)} -1 &\leq \sum_m q_m \sqrt{(d-1)^2   - (d p_{\guess}(X|M=m)-1)^2 } \\
&\leq  \sqrt{(d-1)^2   - \sum_m q_m \Big(d p_{\guess}(X|M=m)-1\Big)^2 }  \\
\label{eqnguessmaxentropyapp5}
&\leq  \sqrt{(d-1)^2   - \bigg(d \sum_m q_m p_{\guess}(X|M=m)-1\bigg)^2 } \,,
\end{align}
where the second (third) inequality uses the concavity (convexity) of the square root (square) function. Combining \eqref{eqnguessmaxentropyapp5} with \eqref{eqnguessmaxentropyapp3b} proves the desired result.
\end{proof}

\section{Asymmetric couplers}\label{app3}

Theorem~\ref{thm111} treats the $n$-path interferometer in Fig.~\ref{fig1} for the special case where $\fc_2$ is \textit{symmetric}, or in other words, \textit{unbiased} with respect to the $n$ output modes. Here we show that this restriction can be relaxed, and one can derive a WPDR for the general scenario where both $\fc_1$ and $\fc_2$ are possibly asymmetric, namely Eq.~\eqref{eqnasymmetricwpdr1}. Furthermore, we will derive \eqref{eqnasymmetricwpdr1} from the MMEUR; more specifically, from our generic relation in \eqref{eqnMainResult}.

First let us introduce our notation.  Recall that $S$ denotes the photon's spatial degree of freedom inside the interferometer and $E$ is the environment. Let $\rho_{SE}$ denote the bipartite state for $S$ and $E$ at a time immediately after these two systems finish interacting inside the interferometer, and denote this time as $t_1$. As shown in Fig.~\ref{fig1}, after time $t_1$, the photon receives a phase shift $\phi_z$ depending on its path, which we can write as the unitary 
\begin{align}
\label{eqnunitaryphase1}
U_{\phiv} := \sum_z e^{i \phi_z}\dya{z}\,.
\end{align}
Then the photon approaches $\fc_2$, whose action is given by some unitary matrix $U_2$ applied to $\HC_S$. Finally the photon is detected at one of the detectors. Let $\Cbb :=\{\Cbb_c\}_{c=1}^n$ denote the POVM associated with detection at the interferometer output. This POVM gives rise to the random variable $C$, which encodes the information about which detector clicks, as noted in the main text. For notational simplicity, we lump together $U_2$ and $\Cbb$ into a single step, and we define the POVM
\begin{align}
\label{eqnlumpedpovm}
\tilde{\Cbb} :=\{\tilde{\Cbb}_c\}_{c=1}^n\,,\quad\text{with}\quad\tilde{\Cbb}_c := U_2\ad \Cbb_c U_2\,.
\end{align}

As noted in the main text, we will need to modify the definitions of distinguishability and visibility in order to formulate a universally valid relation. Let us first consider distinguishability. Consider a game where one tries to guess which path the photon took given that one knows that the photon was detected at a particular detector, say detector $D_1$ (which corresponds to $C=1$). The guessing probability for this game can be written as $p_{\guess}(Z|E,C=1)$, i.e., the probability for guessing the path $Z$ given the optimal measurement on the environment $E$, and given the outcome $C=1$. We define the \textit{post-selected} distinguishability $\DC_1$ as
\begin{align}
\label{eqnPostselectDistDef}
\DC_1 := \frac{n \hspace{2pt}p_{\guess}(Z|E,C=1)-1}{n-1}\,.
\end{align}

Moving on to visibility, we define
\begin{align}
\label{eqnPostselectVisDef}
\VC_1 := \frac{p_{C=1}^{\max} - p_{C=1}^{\text{dec}}}{(n-1)p_{C=1}^{\text{dec}}}\,,
\end{align}
where $p_{C=1}^{\max} := \max_{\phiv} (p_{C=1} )$ is the probability for detector $D_1$ to click maximized over all phases $\phiv$ applied inside the interferometer. Also, $p_{C=1}^{\text{dec}}$ denotes the probability for detector $D_1$ to click in the case where the state inside the interferometer is fully decohered, i.e., where all of the which-path information has leaked out to the environment and hence the system's density matrix has been diagonalized in the which-path basis. Mathematically, we can write
\begin{align}
\label{eqnPostselectVisDef2}
p_{C=1}^{\max} = \max_{\phiv} \Tr \bigg( U_{\phiv}\hspace{2pt}\rho_S \hspace{2pt} U_{\phiv}\ad \hspace{2pt} \tilde{\Cbb}_1 \bigg)\,,\quad \text{and}\quad p_{C=1}^{\text{dec}} = \Tr \bigg( \sum_z \dya{z}\hspace{2pt}\rho_S \hspace{2pt} \dya{z} \hspace{2pt} \tilde{\Cbb}_1 \bigg) \,.
\end{align}

The intuition behind the formula for $\VC_1$ is the following. It quantifies the degree to which adjusting the phases $\phiv$ can increase the probability $p_{C=1}$ to detect the photon at detector $D_1$, beyond the baseline value $p_{C=1}^{\text{dec}}$ associated with no coherence. So it quantifies the effect of the applied phases on the final detection probability. One may notice that $\VC_1$ looks somewhat similar to \eqref{eqnVisDef1}, which is the most common way of writing $\VC$ in the two-path case. Indeed one has $\VC_1 = \VC$ in the two-path case, which can be seen using the following identity \cite{Coles2014b}
\begin{align}
\label{eqnrelationpdecpmax}
p_{C=1}^{\text{dec}} = \big(p_{C=1}^{\max}+p_{C=1}^{\min}\big)/2\,,
\end{align}
which holds in the special case of $n=2$. However, for $n>2$, $\VC_1$ is generally not equal to the expression in \eqref{eqnVisDef1}.

We remark that $\DC_1$ and $\VC_1$ can be experimentally measured as follows. The experimenter, whom we call Alice, can insert a variable attenuator into each arm of the interferometer in Fig.~\ref{fig1}, such that the attenuator can be set to either allow the photon to pass or to block the photon. To measure $\DC_1$, Alice flips a fair $n$-sided classical coin in order to determine which path inside the interferometer to keep open (to allow the photon to pass), while blocking the other $n-1$ paths. To compute $p_{\guess}(Z|E,C=1)$ Alice employs a second experimenter, named Bob, who tries to guess which path Alice kept open, given that Bob has access to $E$ and given that he post-selects on $C=1$ outcomes. In the case of $\VC_1$, measuring $p_{C=1}^{\max}$ is straightforward, while $p_{C=1}^{\text{dec}}$ can be measured as follows. Alice once again flips a fair $n$-sided classical coin to determine which path inside the interferometer to keep open (while blocking the other $n-1$ paths), and then computes $p_{C=1}^{\text{dec}}$ as the number of detection events at $D_1$ divided by the total number of detection events.

Ref.~\cite{Li2012} proved a WPDR for the two-path MZI for the general case where both fiber couplers (or beam splitters in the case of free-space propagation) are asymmetric. We now state the following theorem, which generalizes the WPDR in Ref.~\cite{Li2012} to the $n$-path MZI.
\begin{theorem}
\label{thm1122}
For the $n$-path MZI in Fig.~\ref{fig1}, where $\fc_1$ and $\fc_2$ are arbitrary (possibly asymmetric) fiber couplers, it holds that
\begin{align}
\label{eqnVisDisTradeoff22241}
\VC_1^2+\DC_1^2 \leq 1\,.
\end{align}
Moreover, \eqref{eqnVisDisTradeoff22241} is a special case of the MMEUR in \eqref{eqnMainResult}. 
\end{theorem}
\begin{proof}
The proof simply involves applying \eqref{eqnMainResult} to the appropriate density operator.

Consider an isometry $V$ that produces a copy of the which-path information and stores it in a register $S'$, which we can write as a map $\HC_{SE}\to  \HC_{S'SE}$ given by
\begin{align}
\label{eqnVisDisTradeoffProof1}
V := \sum_z \ket{z}_{S'} \ot\dya{z}_S \ot \id_E \,.
\end{align}
Consider the density operator obtained from applying this isometry to $\rho_{SE}$,
\begin{align}
\label{eqnVisDisTradeoffProof2}
\rhot_{S'SE}:=V\rho_{SE}V\ad \,.
\end{align}
We note that $\rhot_{S'SE}$ is not actually the physical state; rather it is a mathematical construction that we use to prove the desired result. 

Now consider the density operator obtained from post-selecting on events where detector $D_1$ clicks (i.e., where $C=1$). That is, we take the density operator $\rhot_{S'SE}$ and we consider its overlap with the POVM element $\tilde{\Cbb}_1$ associated with the event $C=1$. Applying this post-selection to $\rhot_{S'SE}$ gives the following density operator,
\begin{align}
\label{eqnVisDisTradeoffProof3}
\rhoh_{S'E}:=\frac{\Tr_S [(\id_{S'} \ot \tilde{\Cbb}_1 \ot \id_{E})\rhot_{S'SE} ]  }{\Tr [(\id_{S'} \ot \tilde{\Cbb}_1 \ot \id_{E})\rhot_{S'SE} ]}\,.
\end{align}

We now apply \eqref{eqnMainResult} to the density operator in \eqref{eqnVisDisTradeoffProof3}, setting $E_2$ to a trivial system and $E_1 = E$, giving
\begin{align}
\label{eqnVisDisTradeoffProof4}
H_{\min}(Z_{S'} | E )_{\rhoh}+ \min_{\Wbb\in \MC_{\Zbb}}H_{\max}(W_{S'}  )_{\rhoh} \geq \log n\,,
\end{align}
where the $\rhoh$ subscripts in \eqref{eqnVisDisTradeoffProof4} serve as a reminder that the entropy terms are evaluated for the state in \eqref{eqnVisDisTradeoffProof3}, and the $S'$ subscripts in \eqref{eqnVisDisTradeoffProof4} indicates that the random variables arise from observables on system $S'$. Noting that $Z_{S'}$ is a copy of the which-path information for system $S$, one can see that 
\begin{align}
\label{eqnVisDisTradeoffProof5}
H_{\min}(Z_{S'} | E )_{\rhoh} &= - \log p_{\guess}(Z_{S'} | E )_{\rhoh}  \\
\label{eqnVisDisTradeoffProof6}
&= - \log p_{\guess}(Z_{S} | E, C=1 )_{\rho} \\
\label{eqnVisDisTradeoffProof7}
&=   - \log \left(\frac{1+(n-1)\DC_1 }{n} \right)\,,
\end{align}
where we note that \eqref{eqnVisDisTradeoffProof6} refers to the physical state $\rho_{SE}$. Likewise we can relate the max-entropy term in \eqref{eqnVisDisTradeoffProof4} to $\VC_1$ as follows
\begin{align}
\label{eqnVisDisTradeoffProof8}
\min_{\Wbb\in \MC_{\Zbb}}H_{\max}(W_{S'}  )_{\rhoh} \leq \log \left(1+(n-1)\sqrt{1-\VC_1^2} \right)\,,
\end{align}
where we elaborate on the proof of \eqref{eqnVisDisTradeoffProof8} below. Inserting \eqref{eqnVisDisTradeoffProof7} and \eqref{eqnVisDisTradeoffProof8} into \eqref{eqnVisDisTradeoffProof4} and rearranging gives the desired result \eqref{eqnVisDisTradeoff22241}.

Eq.~\eqref{eqnVisDisTradeoffProof8} is proved as follows. First, using \eqref{eqnguessmaxentropy} gives
\begin{align}
\label{eqnVisDisTradeoffProof9}
\min_{\Wbb\in \MC_{\Zbb}}H_{\max}(W_{S'}  )_{\rhoh} & \leq \log \left(1+(n-1) \sqrt{1- \left(\frac{n \max_{\Wbb\in \MC_{\Zbb}} p_{\guess}(W_{S'}  )_{\rhoh}-1  }{n-1}\right) ^2   } \right)\,.
\end{align}
So it remains only to show that
\begin{align}
\label{eqnVisDisTradeoffProof10}
n \max_{\Wbb\in \MC_{\Zbb}} p_{\guess}(W_{S'}  )_{\rhoh}  = \frac{p_{C=1}^{\max}  }{ p_{C=1}^{\text{dec}}}\,.
\end{align}
The proof of \eqref{eqnVisDisTradeoffProof10} goes as follows,
\begin{align}
\label{eqnVisDisTradeoffProof11}
 \max_{\Wbb\in \MC_{\Zbb}} p_{\guess}(W_{S'}  )_{\rhoh}  & =  \max_{\Wbb\in \MC_{\Zbb}} p_{\guess}(W_{S'} | C=1 )_{\rhot}\\
\label{eqnVisDisTradeoffProof12} & =  \max_{\Wbb\in \MC_{\Zbb}} \max_{w} \Pr(W_{S'} = w | C=1 )_{\rhot}\\
\label{eqnVisDisTradeoffProof13}& =  \max_{\Wbb\in \MC_{\Zbb}} \max_{w} \frac{\Pr(W_{S'} = w \hspace{2pt} , \hspace{2pt} C=1 )_{\rhot} }{\Pr (C=1)_{\rhot}  }\\
\label{eqnVisDisTradeoffProof14}& =  \max_{\Wbb\in \MC_{\Zbb}} \max_{w} \frac{  \Tr (( \dya{\Wbb_w}\ot \tilde{\Cbb}_1) \rhot_{S'S} )}{  \Tr ( \tilde{\Cbb}_1  \rhot_{S} )  }\\
\label{eqnVisDisTradeoffProof15}& =  \frac{1}{p_{C=1}^{\text{dec}}} \max_{\Wbb\in \MC_{\Zbb}} \max_{w} \Tr (( \dya{\Wbb_w}\ot \tilde{\Cbb}_1) \rhot_{S'S} )  \\
\label{eqnVisDisTradeoffProof16}& =  \frac{1}{p_{C=1}^{\text{dec}}} \max_{\{\phi_z\}} \Tr \left(\left( \sum_{z,z'} \frac{e^{i(\phi_z - \phi_{z'})}}{n}\dyad{z}{z'} \ot \tilde{\Cbb}_1 \right) \rhot_{S'S} \right)  \\
\label{eqnVisDisTradeoffProof17}& =  \frac{1}{n p_{C=1}^{\text{dec}}} \max_{\{\phi_z\}} \Tr \left(\left( \sum_{z,z'} e^{i(\phi_z - \phi_{z'})} \dyad{z}{z'} \ot \tilde{\Cbb}_1 \right) \left( \sum_{z'',z'''} \dyad{z''}{z'''} \ot \dya{z''}\rho_S\dya{z'''}\right) \right)  \\
\label{eqnVisDisTradeoffProof18}& =  \frac{1}{n p_{C=1}^{\text{dec}}} \max_{\{\phi_z\}} \Tr \left( \sum_{z,z'} e^{i(\phi_z - \phi_{z'})} \tilde{\Cbb}_1  \dya{z'}\rho_S\dya{z} \right)  \\
\label{eqnVisDisTradeoffProof19}& =  \frac{1}{n p_{C=1}^{\text{dec}}} \max_{\{\phi_z\}} \Tr \left(  \tilde{\Cbb}_1  U_{\phiv}\rho_S U_{\phiv}\ad \right)  \\
& =  \frac{p_{C=1}^{\max} }{n p_{C=1}^{\text{dec}}}  \,.
\end{align}
Here, $\Pr$ denotes ``probability'', and \eqref{eqnVisDisTradeoffProof13} uses the fact that the joint probability for two events $X=x$ and $Y=y$ is given by $\Pr (X=x,Y=y) = \Pr (X=x | Y=y) \Pr(Y=y)$. Also, \eqref{eqnVisDisTradeoffProof16} uses the fact that any state $\ket{\psi}$ that is unbiased to the $\Zbb$ basis can be written as $\ket{\psi}=\sum_z (e^{i\phi_z}/\sqrt{n})\ket{z}$ for some set of phases $\{\phi_z\}$. This concludes the proof.
\end{proof}

\section{Quantum erasure}\label{app2}

Here we give a proof of \eqref{eqnqerase123}.  As a reminder, we restate this result as follows. Recall that a quantum erasure experiment involves performing a measurement $\Ybb = \{\Ybb_y\}$ on the environment $E$ and using the measurement outcome to sort the data into sub-ensembles. Let
\begin{equation}
\label{eqndefdyvyapp1}
\DC(\Ybb_y):= \frac{n \hspace{2pt}p_{\guess}(Z| Y=y) -1}{n-1}\quad \text{and}\quad \VC(\Ybb_y):= \frac{n \max_{\phiv}p_{\guess}(C| Y=y) -1}{n-1}
\end{equation}
denote the path distinguishability and fringe visibility, respectively, for the $y$-th sub-ensemble, i.e., associated with outcome $y$ of the $\Ybb$ measurement. Averaging over all $y$ gives the quantities
\begin{align}
\label{eqndefdyvyapp2}
\DC(\Ybb):= \sum_y p_y \DC(\Ybb_y)\quad \text{and}\quad\VC(\Ybb):= \sum_y p_y \VC(\Ybb_y)
\end{align}
where $p_y$ is the probability of outcome $y$. Then we have the following result, which is a generalization of Theorem~\ref{thm111}. (We remark that, while the following theorem is stated for the case where $\fc_2$ is symmetric, this assumption can be relaxed using the approach in Appendix~\ref{app3}, where the definitions of visibility and distinguishability are slightly modified.)
\begin{theorem}
\label{thmqerase1}
For the $n$-path MZI in Fig.~\ref{fig1}, where $\fc_1$ is arbitrary while $\fc_2$ is symmetric, it holds that
\begin{align}
\label{eqnVisDisTradeoff2224112}
\VC(\Ybb)^2+\DC(\Ybb)^2 \leq 1\,.
\end{align}
\end{theorem}
\begin{proof}
A quantum erasure experiment can be separated into three sequential steps, each of which is a CPTP map:
\begin{enumerate}
  \item The system $S$ (i.e., the photon's spatial degree of freedom) interacts with an environment $\En$, via CPTP map $\EC_{\text{int}}$.
  \item System $\En$ is measured and the outcome is stored in a register $\register$, via CPTP map $\EC_{\text{meas}}$.
  \item The experimenter uses this measurement result to enhance the visibility on system $S$ (i.e., to sort the data point into a sub-ensemble and determine the optimal phase shift for that sub-ensemble). This is modelled as a CPTP map $\EC_{\enh}$ that couples $\register$ to $S$.
\end{enumerate}
The overall CPTP map $\EC$ is a composition of these three maps:
\begin{equation}
\EC = \EC_{\enh} \circ \EC_{\text{meas}} \circ \EC_{\text{int}}. 
\end{equation}
Suppose the system starts in the state $\rho_{S}^{(1)}$, just after it enters the interferometer. The interaction with $\En$ results in the state $\rho_{S\En}^{(2)} := \EC_{\text{int}}(\rho_{S}^{(1)})$. Next, $\EC_{\text{meas}}$ performs the measurement $\Ybb = \{\Ybb_{\outcome}\}$ on system $\En$ and stores the outcome in two (redundant) registers $\register$ and $\register '$, resulting in the state
\begin{equation}
\rho_{S\register \register '}^{(3)}:= \EC_{\text{meas}}(\rho_{S\En}^{(2)}) = \sum_{\outcome} p_{\outcome} \rho_{S,y}^{(2)} \ot \dya{\outcome}_\register \ot \dya{\outcome}_{\register '}\,.
\end{equation}
Here, state $\ket{\outcome}$ corresponds to obtaining outcome $\outcome$ from measuring $\Ybb$, with the set $\{\ket{\outcome} \}$ forming an orthonormal basis on the register Hilbert space, and outcome $\outcome$ leaves system $S$ in the conditional state 
\begin{align}
\rho_{S, \outcome}^{(2)} = \frac{1}{p_{\outcome}}\Tr_{\En} \big[ (\id_S \ot \Ybb_{\outcome}) \rho_{S \En}^{(2)} \big],\quad\text{with}\quad p_{\outcome} =\Tr \big[ (\id_S \ot \Ybb_{\outcome})  \rho_{S \En}^{(2)} \big] .
\end{align}
The point of having two registers ($\register$ and $\register '$) is that one register will act as system $E_1$ - to be used to enhance the distinguishability - while the other will act as system $E_2$ - to be used to enhance the visibility.

For each outcome $\outcome$, we wish to obtain the full visibility that is available. So we allow the experimenter to choose the optimal set of phase shifts $\phiv_{\outcome}$ for each $\outcome$, i.e., let $\phiv_{\outcome}$ achieve the optimization in $\VC(\Ybb_{\outcome})$ as defined in \eqref{eqndefdyvyapp1}. In other words, we allow the experimenter, given the outcome $\outcome$, to rotate the system $S$ via a unitary $U_{\phiv_{\outcome}}$ that has the form given in \eqref{eqnunitaryphase1}. Accounting for all possible values of $\outcome$, the overall unitary is a controlled unitary $U_{\enh} := \sum_{\outcome} U_{\phiv_{\outcome}} \ot \dya{\outcome}_{\register}$ where $\register$ acts as the control system. Hence the action of the map that enhances the visibility is
\begin{align}
\label{eqnrho4qerase}
\rho_{\sys \register \register'}^{(4)}&:= \EC_{\enh}(\rho_{\sys \register \register '}^{(3)}) = U_{\enh} \rho_{\sys \register \register '}^{(3)} U_{\enh}\ad = \sum_{\outcome} p_{\outcome} \rhot_{\sys , \outcome}^{(2)}  \ot \dya{\outcome}_{\register} \ot \dya{\outcome}_{\register '}\,,\quad \text{with}\quad \rhot_{\sys , \outcome}^{(2)}:= U_{\phiv_{\outcome}}  \rho_{\sys ,\outcome}^{(2)} \Big(U_{\phiv_{\outcome}}  \Big)\ad \,.
\end{align}
Applying \eqref{eqnMainResult} to the state $\rho_{\sys \register \register'}^{(4)}$, choosing $E_1=\register$ and $E_2 =\register ' $, and noting that $\register '$ and $\register$ are identical copies and hence can be interchanged, gives
\begin{equation}
\label{eqnQuEr1111}
H_{\min}(Z|\register )_{\rho^{(4)}}+\min_{W\in \MC_{\Zbb}}H_{\max}(W|\register )_{\rho^{(4)}} \geq \log n\,,
\end{equation}
where the subscript $\rho^{(4)}$ emphasizes that the entropy terms are evaluated for the state $\rho_{\sys \register \register '}^{(4)}$. Equation~\eqref{eqnQuEr1111} is our quantum erasure relation, corresponding to Eq.~\eqref{eqnQuEr111122} in the main text.

We note that \eqref{eqnQuEr1111} can be rewritten as follows. Using the definition of $\DC(\Ybb)$ in \eqref{eqndefdyvyapp2}, we have
\begin{align}
H_{\min}(Z|\register )_{\rho^{(4)}} &= - \log p_{\guess}(Z|\register ) \\
&= - \log \left( \sum_{\outcome} p_y p_{\guess}(Z|\register = \outcome )\right) \\
\label{eqnhminerasure}&= \log n - \log ( (n-1)\DC(\Ybb)+1)\,.
\end{align}
Likewise, the definition of $\VC(\Ybb) $ gives
\begin{align}
\label{eqnhmaxerasure1}\min_{\Wbb\in \MC_{\Zbb}}H_{\max}(W|\register )_{\rho^{(4)}} &\leq H_{\max}( C | \register )_{\rho^{(4)}}\\
\label{eqnhmaxerasure2}&\leq \log \left(1+\sqrt{(n-1)^2 - (n \hspace{2pt}p_{\guess}( C | \register )_{\rho^{(4)}}-1)^2   }\right)\\
\label{eqnhmaxerasure3}&= \log \left(1+(n-1)\sqrt{1 -  \VC(\Ybb )^2   }\right)\,.
\end{align}
Here, \eqref{eqnhmaxerasure1} uses the fact that $\fc_2$ is symmetric, and hence the POVM $\tilde{\Cbb}$ defined in \eqref{eqnlumpedpovm} corresponds to an orthonormal basis that is mutually unbiased to $\Zbb$, i.e., $\tilde{\Cbb} \in \MC_{\Zbb}$. Equation~\eqref{eqnhmaxerasure2} uses \eqref{eqnguessmaxentropy}. Equation~\eqref{eqnhmaxerasure3} uses the fact that the state $\rho_{\sys \register \register'}^{(4)}$ in \eqref{eqnrho4qerase} is constructed such that $p_{\guess}( C |Y=y)$ is optimal, i.e.,
\begin{align}
\VC(\Ybb) &= \sum_y p_y \VC(\Ybb_y) \\
&=\sum_y p_y \left(\frac{n \hspace{2pt}\max_{\phiv}p_{\guess}(C|Y=y)-1}{n-1}\right)\\
&=\sum_y p_y \left(\frac{n \hspace{2pt}p_{\guess}( C |Y=y)_{\rho^{(4)}}-1}{n-1}\right)\,.
\end{align}
Inserting the expressions in \eqref{eqnhminerasure} and \eqref{eqnhmaxerasure3} into \eqref{eqnQuEr1111} and rearranging gives the desired result \eqref{eqnVisDisTradeoff2224112}.
\end{proof}

\end{document}